\documentclass[runningheads]{llncs}
\usepackage[T1]{fontenc}
\usepackage{orcidlink}
\usepackage{graphicx}
\usepackage[utf8]{inputenc} 
\usepackage{url}
\usepackage{ifthen}
\usepackage{datetime}
\usepackage{epsfig}
\usepackage{pict2e}
\usepackage{color}
\usepackage{mathtools, amsfonts, amssymb, xfrac, bm}
\usepackage{dsfont}
\usepackage{commath}
\usepackage{flushend}

\usepackage[numbers]{natbib}

\newcommand{\A}{\mathcal{A}}
\newcommand{\R}{\mathcal{R}}
\newcommand{\V}{\mathcal{V}}
\newcommand{\W}{\mathcal{W}}
\newcommand{\Z}{\mathcal{Z}}
\renewcommand{\S}{\mathcal{S}}
\newcommand{\QQ}{\mathbb{Q}}
\newcommand{\NN}{\mathbb{N}}
\newcommand{\RR}{\mathbb{R}}
\newcommand{\Lp}{\mathcal{L}_{\mathrm{sup}}}
\newcommand{\Lb}{\mathcal{L}_{\mathrm{sub}}}
\newcommand{\WQ}{\mathcal{W}_{\mathbb{Q}}}
\newcommand{\dr}{\bm{\mathrm{r}}}
\newcommand{\dw}{\bm{\mathrm{w}}}
\newcommand{\dz}{\bm{\mathrm{z}}}
\newcommand{\da}{\bm{\mathrm{a}}}
\renewcommand{\vec}[1]{\overline{#1}}
\newcommand{\cev}[1]{\underline{#1}}
\newcommand{\rr}{\vec{r}}
\newcommand{\lr}{\hspace{1pt}\cev{r}}
\newcommand{\rw}{\vec{w}}
\newcommand{\lw}{\cev{w}}
\newcommand{\rz}{\vec{z}}
\newcommand{\lz}{\hspace{1pt}\cev{z}}
\newcommand{\ra}{\vec{a}}
\newcommand{\la}{\hspace{1pt}\cev{a}}
\newcommand{\rkap}{\vec{\kappa}}
\newcommand{\lkap}{\hspace{1pt}\cev{\kappa}}
\newcommand{\rkapp}{\mathrlap{\phantom{\kappa}^{\hspace{1pt}\prime}}\vec{\kappa}}
\newcommand{\lkapp}{\mathrlap{\phantom{\kappa}^{\hspace{1pt}\prime}}\cev{\kappa}}
\newcommand{\rrp}{\mathrlap{\phantom{r}^{\hspace{1pt}\prime}}\rr}
\newcommand{\rrpx}[1]{\mathrlap{\phantom{r}^{\hspace{1pt}\prime}}\rr_{#1}}
\newcommand{\rrppx}[1]{\mathrlap{\phantom{r}^{\hspace{1pt}\prime\prime}}\rr_{#1}}
\newcommand{\lrp}{\mathrlap{\phantom{r}^{\hspace{1pt}\prime}}\lr}
\newcommand{\lrppx}[1]{\mathrlap{\phantom{r}^{\hspace{1pt}\prime\prime}}\lr_{#1}}
\newcommand{\lrmn}{\mathrlap{\phantom{r}^{\hspace{1pt}m}}\lr_n\hspace{2pt}}
\newcommand{\rrmn}{\mathrlap{\phantom{r}^{\hspace{1.5pt}m}}\rr_n\hspace{2.5pt}}
\newcommand{\rrmnn}{\mathrlap{\phantom{r}^{\hspace{1.5pt}m}}\rr_{n+1}}
\newcommand{\rrmx}[1]{\mathrlap{\phantom{r}^{\hspace{1.5pt}m}}\rr_{#1}}
\newcommand{\sublkap}{\mathrlap{\reflectbox{\ensuremath{\scriptstyle\mathchar"017E}}}\scriptstyle\kappa}
\newtheorem{Theorem}{Theorem}
\newtheorem{Definition}[Theorem]{Definition}
\newtheorem{Lemma}[Theorem]{Lemma}
\newtheorem{Corollary}[Theorem]{Corollary}
\newtheorem{Remark}[Theorem]{Remark}

\begin{document}
\title{On Effective Convergence in Fekete's Lemma and Related Combinatorial Problems in Information Theory}
\titlerunning{On Effective Convergence in Fekete's Lemma}
%
\author{Holger Boche\inst{1,3,4,5}\orcidlink{0000-0002-8375-8946} \and
Yannik Böck\inst{1}\orcidlink{0000-0001-7640-6988} \and
Christian Deppe\inst{2,3}\orcidlink{0000-0002-2265-4887}}
\authorrunning{H. Boche, Y. Böck, C. Deppe}
%
\institute{Technical University of Munich, TUM School of Computation, Information and Technology, Munich, Germany \and
Technical University of Braunschweig, Institute for Communications Technology, Braunschweig, Germany \and 6G-life, 6G research hub, Germany
\and Munich Quantum Valley (MQV), Leopoldstra{\ss}e 244, 80807 M\"unchen, Germany
\and Munich Center for Quantum Science and Technology, Schellingstra{\ss}e 4, 80799 München, Germany\\
\email{boche@tum.de, yannik.boeck@tum.de, christian.deppe@tu-braunschweig.de}}
\maketitle              
\centerline{\bf In Memory of Ning Cai}

\sloppy

\begin{abstract}
Fekete's lemma is a well known result from combinatorial mathematics that shows the existence of a limit value related to super- and subadditive sequences of real numbers.
In this paper, we analyze Fekete's lemma in view of the arithmetical hierarchy of real numbers by \emph{\citeauthor{ZhWe01}} and fit the results into an information-theoretic context. We introduce special sets associated to super- and subadditive sequences and prove their effective equivalence to \(\Sigma_1\) and \(\Pi_1\). Using methods from the theory established by \emph{\citeauthor{ZhWe01}}, we then show that the limit value emerging from Fekete's lemma is, in general, not a computable number. Given a sequence that additionally satisfies non-negativity, we characterize under which conditions the associated limit value can be computed effectively and investigate the corresponding modulus of convergence. Subsidiarily, we prove a theorem concerning the structural differences between computable sequences of computable numbers and computable sequences of rational numbers. We close the paper by a discussion on how our findings affect common problems from information theory.
\end{abstract}

\section{Introduction}                              %
\label{sec:Intr}                                    %

    A sequence \((a_n)_{n\in\NN}\) of real numbers is called \emph{superadditive} if it satisfies \(a_{n+m} \geq a_n + a_m\) for all \(n,m \in \NN\). It is called \emph{subadditive} if it satisfies
    \(a_{n+m} \leq a_n + a_m\) for all \(n,m \in \NN\). Fekete's lemma proves the existence of a limit value related to sequences that are either super- or subadditive:

    \begin{Lemma}[Fekete \cite{F23}]\label{lem:fekete}
        \begin{itemize}
            \item[1.] 
                Let \((\ra_n)_{n\in\NN}\) be a sequence of real numbers that satisfies superadditivity. Then, there exists \(x_* \in \RR\cup \{+\infty\}\) such that 
                    \begin{align}
                        \sup_{n\in\NN} \frac{\ra_n}{n}
                        = \lim_{n\to\infty} \frac{\ra_n}{n}
                        = x_*
                    \end{align}
                holds true.
            \item[2.] 
                Let \((\la_n)_{n\in\NN}\) be a sequence of real numbers that satisfies subadditivity. Then, there exists \(x_* \in \RR\cup \{-\infty\}\) such that 
                    \begin{align}
                        \inf_{n\in\NN} \frac{\la_n}{n}
                        = \lim_{n\to\infty} \frac{\la_n}{n}
                        = x_*
                    \end{align}
                holds true.
        \end{itemize}
    \end{Lemma}

    Since the works of Bruijn and Erd\"os \cite{BE51,BE52}, several authors have developed generalizations of Fekete's lemma (see \cite{BR06}). For an up-to-date overview and interesting discussions, we refer to the work of Füredi and Ruzsa \cite{FR18}. In the present paper, we aim to characterize the computability-properties of the limit value by applying a theory established by \emph{\citeauthor{ZhWe01}} \cite{ZhWe01}, known as  \emph{the arithmetical hierarchy of real numbers}.

    In addition to common problem statements from information theory, which we will treat in detail throughout the paper, super- and subadditivity of sequences of real numbers occurs throughout different other areas of mathematics. In economics, it is an essential property of some cost functions. Similar relations also appear in physics and in combinatorial optimization (see \cite{S97}). Further applications of Fekete's lemma exist in the research of cellular automatas \cite{C08} and for analogues on cancellative amenable semigroups \cite{K12}. Another important example can be found in Ramsey theory. The proof of Roth's fundamental theorem \cite{R52} also makes use of Fekete's lemma. Thus, our results are of high relevance for a great variety of research areas.

    The majority of problems from information theory involve the calculation of \emph{channel capacities}. In most cases, a certain channel model gives rise to a monotonically non-decreasing sequence \((M_n)_{n\in\NN}\) of natural numbers, where \(n\) corresponds to the number of channel uses and \(M\) corresponds to the number of messages transmittable through \(n\) uses of the channel. Subsequently, the quantities \(\liminf_{n\to\infty} \sfrac{1}{n}\cdot\log_2 M_n\) and \(\limsup_{n\to\infty} \sfrac{1}{n}\cdot\log_2 M_n\), which are referred to as \emph{pessimistic} and \emph{optimistic} capacity, respectively, are of interest. The mapping \(n\mapsto M_n\) is fully characterized by a sequence of \emph{optimal codes} (with respect to a fixed error constraint) for the channel under consideration, since each optimal code determines the maximal number of messages transmittable given a certain number of channel uses. Assuming finite input and output alphabets, the set of possible codes is finite for each number of channel uses \(n\in\NN\). Thus, the construction of suitable codes is a purely combinatorial task, and a great deal of research has been conducted in this area. For a comprehensive introduction to combinatorial methods in information theory, see \cite{CsKo11}. Surprisingly, the mapping \(n \mapsto M_n\) does not have to be a recursive function. This observation, which we will discuss further in \emph{Remark \ref{RfooI}}, yields an interesting relation between information theory, combinatorics and the theory of computability.

    In general, optimistic and pessimistic capacity are not equal. However, some channel models exhibit a mathematical structure known as \emph{achievability and converse}. Essentially, this structure is characterized by a pair \(\big((\rw_n)_{n\in\NN},(\lw_n)_{n\in\NN}\big)\) of monotonic sequences that satisfy {\color{black}
    \begin{align}
        \inf_{n\in\NN} \lw_n    &\leq \liminf_{n\to\infty} \frac{1}{n}\log_2 M_n \nonumber\\
                                &\leq \limsup_{n\to\infty}  \frac{1}{n}\log_2 M_n \nonumber\\
                                &\leq \sup_{n\in\NN} \rw_n
    \end{align}
    as well as
    \begin{align}
        \sup_{n\in\NN} \rw_n \leq \inf_{n\in\NN} \lw_n,    
    \end{align}}
    in which case the quantity
    \begin{align}
        \liminf_{n\to\infty} \frac{1}{n}\log_2 M_n 
            &= \limsup_{n\to\infty} \frac{1}{n}\log_2 M_n \nonumber\\
            &= \lim_{n\to\infty}\frac{1}{n}\log_2 M_n    
    \end{align}
    is referred to as the \emph{capacity} of the channel (the sequence \((\lw_n)_{n\in\NN}\) will be referred to as "the converse" in this context). This case is particulary desirable with respect to computability, as we will discuss in \emph{Section \ref{converge}}. For a concise discussion on the approach of achievability and converse in information theory, we refer to \cite{Ah06}.

    A second possibility of proving the existence of a well-defined capacity is by the application of Fekete's lemma, in which case the equality of pessimistic and optimistic capacity follows without achievability and converse. However, the applicability of Fekete's Lemma highly depends on the individual channel model. In particular, it is necessary to prove that the sequence \((\sfrac{1}{n}\cdot\log_2 M_n)_{n\in\NN}\) satisfies superadditivity. The arguably most prominent example is Shannon's zero-error capacity, where the number \(M_n\) equals the independence number of the \(n\)th tensor-power of a specific simple graph:\pagebreak
    \begin{Theorem}[Shannon \cite{S56}]\label{Shannon}
        Denote by \(C_0(W)\) the \emph{zero-error capacity} of the discrete, memoryless channel \(W\), \(\Theta(G)\) the \emph{Shannon capacity} of the \emph{confusability graph} \(G\) corresponding to \(W\), and \(\alpha(G^{\boxtimes n})\) the \emph{independence number} of the \(n\)-fold strong graph product of \(G\) with itself. Then, we have
     	\begin{align}
     	     \Theta(G) :    &=\liminf_{n \to \infty}\alpha(G^{\boxtimes n})^{\frac 1n}\nonumber\\ 
     	                    &= \sup_{n \in \NN}\alpha(G^{\boxtimes n})^{\frac 1n}\nonumber\\
     	                    &= \lim_{n \to \infty} \alpha(G^{\boxtimes n})^{\frac 1n}\nonumber\\ 
     	                    &= 2^{C_0(W)} .
    	\end{align}
    \end{Theorem}

    For a number of other channel models, cf. e.g. \cite{ABBN13,BBS13, WNB16}, the application of Fekete's lemma leads to an multi-letter description of the respective capacity. To the authors' knowledge, it remains an open question whether single-letter descriptions for the capacities in \cite{ABBN13,BBS13, WNB16} exist. Given a fixed, computable channel, the descriptions in \cite{ABBN13,BBS13, WNB16} provide a computable sequence \((\sfrac{1}{n}\cdot\log_2 M_n)_{n\in\NN}\) of computable numbers that meet the requirements of Fekete's lemma. Up to now, however, it remains unknown if the respective capacities always attain computable values. In other words, it is not yet known whether these capacites can always be algorithmically computed, e.g. as a floating point number, up to any specified precision. As a corollary of our results, we obtain that in contrast to the achievability and converse approach, Fekete's lemma is structurally insufficient to provide a computable capacity.

    Finding computable expressions in information theory, such as the results in \cite{ABBN13}, \cite{BBS13} and \cite{WNB16}, was a central concern of Ning Cai. Ning Cai considered characterizing capacities as limits of multi-letter representations only as a preliminary result. Ning Cai was always looking for approaches to characterize the convergence speed algorithmically. Such general techniques are still unknown in information theory. In this paper we show that such techniques cannot exist for the Fekete lemma.

    Of course, multi-letter capacity formulae are also very helpful for practical applications. For example, practically relevant questions for the wiretap channel such as continuity and additivity were answered in \cite{boche2013capacity} and \cite{boche2015continuity}. This work was initiated by Ning Cai's questions and until today there are no general approaches in information theory to directly answer questions about continuity and additivity in an elementary way. Rudolf Ahlswede and Rheinhard Werner independently posed the question of the existence of such elementary direct approaches.

    Since \citeauthor{T36} published his work on the theory of computation \cite{T36, T36b}, it has been known that almost all real numbers are uncomputable. In an attempt to characterize different degrees of (un)computability, \emph{\citeauthor{ZhWe01}} \cite{ZhWe01} introduced the arithmetical hierarchy of real numbers, which is strongly related to the \emph{Kleene–Mostowski hierarchy} of subsets of natural numbers \cite{Kl43, Mo47}. In order for a real number \(a\in\RR\) to be computable, it has to be on the lowest Zheng-Weihrauch hierarchical level. The hierarchical level of a real number solely depends on the logical structure that is used to define the number. With respect to channel capacities, this logical structure is given by the corresponding channel model. In information theory, or, in the broadest sense, any mathematical theory that investigates asymptotic quantities, a comprehensive characterization of the hierarchical levels induced by the logical structures involved may be desirable for several reasons. Above all, it is necessary in order to be able to decide which of the asymptotic quantities can be numerically approximated in a feasible manner for practical purposes. Our work aims to contribute to this characterization in view of Fekete's lemma. In particular, we show that a whole class of uncomputable problems can be represented in terms of sub- and superadditive sequences of real numbers.

    The outline of the remainder of our work is as  follows. \emph{Section \ref{basic}} yields an overview on the fundamentals of computability theory, which will be applied subsequently. In \emph{Section \ref{converge}}, we will give a brief introduction on the sets \(\Sigma_1 \subset \RR\) and \(\Pi_1 \subset \RR\), which form the lowest Zheng-Weihrauch hierarchical level, and present the first part of our results: \emph{Theorems \ref{thm:CompleteMonRepr}} and \emph{\ref{T1.5}} address the problem of determining the modulus of convergence for monotonic sequences, as they occur in the information-theoretic achievability and converse approach. In \emph{Section \ref{const}}, we place Fekete's lemma in the context of the arithmetical hierarchy of real numbers. In particular, we prove in \emph{Theorem \ref{thm:LsubLsupRc}} an effective equivalence between \(\Sigma_1\)(\(\Pi_1\)) and the set of real numbers that can be charcaterized by a super(sub)additive sequence through Fekete's lemma. As an immediate consequence (\emph{Corollary \ref{C1.1}}), we obtain that convergence related to Fekete's lemma is generally ineffective. Several implications of this issue are discussed in \emph{Remark \ref{RfooI}} and \emph{Remark \ref{non}}. In \emph{Section \ref{effect}}, we characterize under which circumstances the convergence related to Fekete's lemma is effective. These results are presented in \emph{Theorem \ref{T1.3}}. \emph{Theorem \ref{T1.7}}, furthermore, addresses the problem of determining a suitable corresponding modulus of convergence. The Section concludes with \emph{Theorem \ref{thm:analogonFekete}}, an analogue to \emph{Theorem \ref{thm:CompleteMonRepr}} in the context of Fekete's lemma. \emph{Section \ref{comments}} discusses a subtle difference in the representation of computable sequences. This difference turns out to have significant implications, as is highlighted by \emph{Theorem \ref{thm:foo}}. The paper closes in \emph{Section \ref{sec:discussions}} with a discussion on the consequences of our findings for modern-day approaches in combinatorial coding and information theory.

\section{Preliminaries from the Theory of Computation}          %
\label{basic}                                                   %

    In order to investigate Fekete's lemma in view of computability, we apply the theory of \emph{Turing machines} \cite{T36, T36b} and \emph{recursive functions} \cite{Kl36}. For brevity, we restrict ourselves to an informal description. A comprehensive formal introduction on the topic may be found in \cite{Sc74, W00, So87, PoRi17}.

    Turing machines are a mathematical model of what we intuitively understand as computation machines. In this sense, they yield an abstract idealization of today's real-world computers. Any algortithm that can be executed by a real-world computer can, in theory, be simulated by a Turing machine, and vice versa. In contrast to real-world computers, however, Turing machines are not subject to any restrictions regarding energy consumption, computation time or memory size. Furthermore, the computation of a Turing machine is assumed to be executed free of any errors.

    Recursive functions, more specifically referred to as \emph{\(\mu\)-recursive functions}, form a special subset of the set \(\bigcup_{n=0}^{\infty} \big\{ f : \NN^{n} \hookrightarrow \NN \big\}\), where we use the symbol "\(\hookrightarrow\)" to denote a \emph{partial mapping}. The set of recursive functions characterizes the notion of computability through a different approach. Turing machines and recursive functions are equivalent in the following sense: a function \(f : \NN^{n} \hookrightarrow \NN\) is computable by a Turing machine if and only if it is a recursive function \cite{Tu37}.

    In the following, we will introduce some definitions from \emph{computable analysis} \cite{W00, So87, PoRi17}, which we will apply subsequently.

    \begin{Definition}\label{ber}
        A sequence of rational numbers \((r_n)_{n\in\NN}\) is said to be  \emph{computable} if there exist recursive functions \(f_{\mathrm{si}},f_{\mathrm{nu}},f_{\mathrm{de}}:\NN\to\NN\) such that
        \begin{align}
            r_n= (-1)^{f_{\mathrm{si}}(n)}\frac {f_{\mathrm{nu}}(n)}{f_{\mathrm{de}}(n)}  
        \end{align}
        holds true for all \(n\in\NN\). A double sequence of rational numbers \((r_{n,m})_{n,m\in\NN}\) is said to be \emph{computable} if there exist recursive functions \(f_{\mathrm{si}},f_{\mathrm{nu}},f_{\mathrm{de}}:\NN\times\NN\to\NN\) such that
        \begin{align}
            r_{n,m}= (-1)^{f_{\mathrm{si}}(n,m)}\frac {f_{\mathrm{nu}}(n,m)}{f_{\mathrm{de}}(n,m)}  
        \end{align}
        holds true for all \(n,m\in\NN\).
    \end{Definition}

    \begin{Definition}\label{defeff}
        A sequence \((x_n)_{n\in\NN}\)  of real numbers is said to converge \emph{effectively} towards a number \(x_*\in\RR\) if there exists a recursive function \(\kappa:\NN\to\NN\) such that \(|x_*-x_n|<\sfrac{1}{2^N}\) holds true for all \(n,N\in\NN\) that satisfy \(n\geq \kappa(N)\).
    \end{Definition}

    The function \(\kappa\) is referred to as (recursive) \emph{modulus of convergence} for the sequence \((x_n)_{n\in\NN}\).

    \begin{Definition}\label{compreal}
        A real number \(x\) is said to be \emph{computable} if there exists a computable sequence of rational numbers that converges effectively towards \(x\). 
    \end{Definition}

    We denote the set of computable real numbers by \(\RR_c\). Given a computable number \(x\), a pair \(\big((r_n)_{n\in\NN},\kappa\big)\) consisting of a computable sequence \((r_n)_{n\in\NN}\) of rational numbers that satisfies \(\lim_{n\to\infty} r_n = x\) and a corresponding recursive modulus of convergence \(\kappa\) is called a \emph{standard description} of the number \(x\). We denote by \(\R_c\) the set of standard descriptions of computable numbers.

    \begin{Remark}\label{rem:RcEffectiveField}
        The set \(\RR_c\) is a subfield of the set \(\RR\). The field properties of \(\RR_c\) are effective in the following sense: there exist Turing machines \(TM_{+} : \R_c \times \R_c \rightarrow \R_c\) and \(TM_{\times} : \R_c \times \R_c \rightarrow \R_c\) such that if \(\dr_{x}\in\R_c\) and \(\dr_{y}\in\R_c\) are standard descriptions of \(x\in\RR_c\) and \(y\in\RR_c\), then \(TM_{+}(\dr_x,\dr_y)\) and \(TM_{\times}(\dr_x,\dr_y)\) are standard descriptions of \(x+y\) and \(x\cdot y\), respectively. Furthermore, the set \(\RR_c\) is countable and closed with respect to effective convergence.
    \end{Remark}

    \begin{Definition}\label{def:CSCN}
        A sequence \((x_n)_{n\in\NN}\) of computable numbers is called \emph{computable} if there exists a computable double sequence \((r_{n,m})_{n,m\in\NN}\) of rational numbers as well as a recursive function \(\kappa : \NN \times \NN \rightarrow \NN\) such that 
        \begin{align}
            |x_n - r_{n,m}| < \frac{1}{2^M}
        \end{align}
        holds true for all \(n,m,M\in\NN\) that satisfy \(m \geq \kappa(n,M)\).
    \end{Definition}

    The pair \(\big((r_{n,m})_{n,m\in\NN}, \kappa\big)\) is referred to as a \emph{standard description} of the sequence \((x_n)_{n\in\NN}\).

    An essential component in proving uncomputability of some kind is the notion of \emph{recursive} and \emph{recursively enumerable} sets, which we will introduce in the following.

    \begin{Definition} 
        A set \(A\subseteq \NN\) is said to be \emph{recursively enumerable} if there exists a recursive function \(f : \NN \hookrightarrow \NN\) with domain equal to \(A\).
    \end{Definition}

    \begin{Definition}
        A set \(A\subseteq \NN\) is said to be \emph{recursive} if the corresponding indicator function \(\mathds{1}_{A} : \NN \rightarrow \{0,1\}\) is a recursive function.
    \end{Definition}

    \begin{Remark}
        A set \(A \subseteq \NN\) is recursive if and only if both \(A\) and \(A^{c} := \NN\setminus A\) are recursively enumerable sets. The \emph{halting problem} for Turing machines ensures the existence of recursively enumerable sets that are nonrecursive.
    \end{Remark}

\section{Characterizing Computable Numbers by Monotonic Sequences}      %
\label{converge}                                                        %

    Throughout this section, we address the description of computable numbers in terms of monotonic sequences, which structurally corresponds to a description in the sense of achievability and converse. 

    A real number \(x_*\) is called \emph{upper semi-computable} if there exist a computable sequence \((\lr_n)_{n\in\NN}\) of rational numbers that converges towards \(x_*\) from above, i.e., \((\lr_n)_{n\in\NN}\) satisfies \(\lr_n \geq \lr_{n+1}\) for all \(n\in\NN\) as well as \(\lim_{n\to\infty} \lr_n = x_*\). 

    Likewise, a real number \(x_*\) is called \emph{lower semi-computable} if there exists a computable sequence \((\rr_n)_{n\in\NN}\) of rational numbers that converges towards \(x_*\) from below, i.e, \((\rr_n)_{n\in\NN}\) satisfies \(\rr_n \leq \rr_{n+1}\) for all \(n\in\NN\) as well as \(\lim_{n\to\infty} \rr_n = x_*\). 

    In the above context, the arrow symbols indicate the "direction" (on the real line) towards which the sequences \((\rr_n)_{n\in\NN}\) and \((\lr_n)_{n\in\NN}\) converge.
    \begin{Definition}\label{def:PiSigma}
        We denote by \(\Pi_1\) the set of upper semi-computable real numbers. Likewise, we denote by \(\Sigma_1\) the set of lower semi-computable real numbers.
    \end{Definition}
    We have \(\RR_c = \Sigma_1 \cap \Pi_1\) as well as \(\RR_c\setminus \Sigma_1 \neq \emptyset\) and \(\RR_c\setminus \Pi_1 \neq \emptyset\) \cite{ZhWe01}. In other words, a real number \(x_*\) is computable if and only if it is both upper and lower semi-computable, and there exist real numbers which satisfy only one of the two conditions. 

    The representation of computable numbers through monotonic sequences of rational numbers extends to computable sequences of computable numbers. That is, a real number \(x_*\) is computable if and only if there exists a computable sequence \((\rw_n)_{n\in\NN}\) of computable numbers that converges monotonically non-decreasingly towards \(x_*\) and a computable sequence \((\lw_n)_{n\in\NN}\) of computable numbers that converges monotonically non-increasingly towards \(x_*\). Subsequently, we will work towards a strengthening of this result. In particular, we want to prove that both \((\rw_n)_{n\in\NN}\) and \((\lw_n)_{n\in\NN}\) converge effectively towards \(x_*\) (convergence towards a computable number is not a sufficient condition for effective convergence, cf. \cite{Vy98}) and present an effective construction of the respective moduli of convergence.

    \begin{Lemma}\label{T1.4}
        Denote by \(\WQ\) the set of pairs \(\big((\rr_n)_{n\in\NN},(\lr_n)_{n\in\NN}\big)\) of computable sequences of rational numbers that satisfy the following:
        \begin{enumerate}
            \item 
                \(\rr_n\leq \rr_{n+1}\) for all \(n\in\NN\);
            \item 
                \(\lr_n\geq \lr_{n+1}\) for all \(n\in\NN\);
            \item 
                there exists \(x_*\in\RR\) with\newline \(\lim_{n\to\infty}\rr_n=\lim_{n\to\infty}\lr_n = x_*.\)
        \end{enumerate}
        There exists a Turing machine \(TM : \WQ \times \NN \rightarrow \NN\) such that 
        \begin{align}
            |x_*-\rr_n|<\frac 1{2^M}
            \quad \wedge \quad 
            |x_*-\lr_n|<\frac 1{2^M}
        \end{align}
        holds true for all \(n\in\NN\) that satisfy \(n\geq TM\big((\rr_n)_{n\in\NN}, (\lr_n)_{n\in\NN},M\big)\).
    \end{Lemma}\begin{proof}
        For \(M\in\NN\) arbitrary, let \(n_0\in\NN\) be the smallest number that satisfies \(|\rr_n-\lr_n|<\sfrac{1}{2^M}\). That is, we define
        \begin{align}
            n_0 := \min\left\{n\in\NN: |\rr_n-\lr_n|<\frac {1}{2^M}\right\}.
        \end{align}
        The sequence \((\lr_n-\rr_n)_{n\in\NN}\) is a computable sequence of rational numbers that converges monotonically non-increasingly to zero. Thus, the mapping 
        \begin{align} \big((\rr_n)_{n\in\NN}, (\lr_n)_{n\in\NN}, M\big) \mapsto n_0
        \end{align}
        is recursive and there exists a Turing machine \(TM : \WQ \times \NN \rightarrow \NN\) that computes the number \(n_0\) in dependence of \((\rr_n)_{n\in\NN}\), \((\lr_n)_{n\in\NN}\) and \(M\). For all \(n\in\NN\) that satisfy \(n \geq n_0\), we have
        \begin{align}
            |\rr_n-\lr_n|=\lr_n-\rr_n\leq |\rr_{n_0}-\lr_{n_0}|<\frac 1{2^M}
        \end{align}
        as well as \(\rr_n \leq x_* \leq \lr_n\). Consequently, 
        \begin{align}
            |x_*-\rr_n|<\frac 1{2^M}
            \quad \wedge \quad
            |x_*-\lr_n|<\frac 1{2^M}
        \end{align}
        holds true for all \(n\in\NN\) that satisfy \(n\geq n_0\). Thus, \(TM\) satisfies the required properties.
    \qed\end{proof}

    In the following, we establish the description of computable numbers in terms of computable, monotonic sequences of computable numbers.

    A computable sequence \((\rw_n)_{n\in\NN}\) of computable numbers is referred to as \emph{lower monotonic representation} of the real number \(x_*\) if it satisfies \(\rw_n\leq \rw_{n+1}\) for all \(n\in\NN\) as well as \(\lim_{n\to\infty}\rw_n= x_*\).

    Likewise, a computable sequence \((\lw_n)_{n\in\NN}\) of computable numbers is referred to as an \emph{upper monotonic representation} of the real number \(x_*\) if it satisfies \(\lw_n\geq \lw_{n+1}\) for all \(n\in\NN\) as well as \(\lim_{n\to\infty}\lw_n= x_*\).

    A pair \(\big((\rw_n)_{n\in\NN},(\lw_n)_{n\in\NN}\big)\) is referred to as a \emph{(complete) monotonic representation} of the real number \(x_*\) if \((\rw_n)_{n\in\NN}\) is a lower monotonic representation of \(x_*\) and \((\lw_n)_{n\in\NN}\) is an upper monotonic representation of \(x_*\). Furthermore, we denote by \(\W\) the set of standard descriptions
    \begin{align}
        \dw = \big((\rr_{n,m})_{n,m\in\NN}, \rkap, (\lr_{n,m})_{n,m\in\NN}, \lkap\big)
    \end{align}
    of (complete) monotonic representations of real numbers.

    \begin{Theorem}\label{thm:CompleteMonRepr}
        There exists a Turing machine \(TM : \W\times \NN \rightarrow \NN\) that satisfies the following:
        \begin{itemize}
            \item 
                If \(\big((\rw_n)_{n\in\NN},(\lw_n)_{n\in\NN}\big)\) is a monotonic representation of \(x_* \in \RR\) and \(\dw\) is a standard description of \(\big((\rw_n)_{n\in\NN},(\lw_n)_{n\in\NN}\big)\), then 
                \begin{align}
                    |x_*-\rw_n|<\frac 1{2^M}
                    \quad \wedge \quad 
                    |x_*-\lw_n|<\frac 1{2^M}
                \end{align}
                holds true for all \(n\in\NN\) that satisfy \(n\geq TM(\dw,M)\). 
        \end{itemize}
    \end{Theorem}\begin{proof}
        Consider the standard description \(\dw = \big((\rr_{n,m})_{n,m\in\NN}, \rkap, (\lr_{n,m})_{n,m\in\NN}, \lkap\big)\) of \(\big((\rw_n)_{n\in\NN},(\lw_n)_{n\in\NN}\big)\) and observe that for all \(n,k,l\in\NN\) that satisfy \(k,l \geq n\), we have
        \begin{gather}
            \rr_{n,\rkap(n,n)}  - \frac{1}{2^n} < \rw_n, \\
            \rw_n \leq \rw_{k} \leq x_* \leq \lw_{l} \leq \lw_n, \\
            \lw_n < \lr_{n,\sublkap(n,n)} + \frac{1}{2^n}.
        \end{gather}
        Define the sequences \((\rrp_n)_{n\in\NN}\) and \((\lrp_n)_{n\in\NN}\) through setting
        \begin{align}
            \rrp_n :&= \max \left\{\rr_{m,\rkap(m,m)} - \frac{1}{2^m} : 1\leq m \leq n \right\}, \\
            \lrp_n :&= \min \left\{\lr_{m,\sublkap(m,m)} + \frac{1}{2^m} : 1\leq m \leq n \right\}
        \end{align}
        for all \(n\in\NN\). Then, for all \(n,k,l\in\NN\) that satisfy \(k,l \geq n\), we have
        \begin{align}\label{eq:foo}
            \big|x_* - \rw_{k} \big| < \big|x_* - \rrp_n \big| ~ \wedge ~ 
            \big|x_* - \lw_{l} \big| < \big|x_* - \lrp_n \big|.
        \end{align}
        Both \((\rrp_n)_{n\in\NN}\) and \((\lrp_n)_{n\in\NN}\) are computable sequences of rational numbers and the mapping \(\dw \mapsto \big((\rrp_n)_{n\in\NN}, (\lrp_n)_{n\in\NN}\big)\) is recursive. Furthermore, the pair \(\big((\rrp_n)_{n\in\NN}, (\lrp_n)_{n\in\NN}\big)\) satisfies the requirements of \emph{Lemma \ref{T1.4}} with \(\lim_{n\to\infty}\rrp_n = \lim_{n\to\infty}\lrp_n = x_* \). Based on \eqref{eq:foo}, the theorem then follows by the application of \emph{Lemma \ref{T1.4}}.
    \qed\end{proof}

    \begin{Remark}
        Recall the standard description \(\dr = \big((r_n)_{n\in\NN},\kappa\big) \in \R_c\) of a computable number \(x_*\). If we can find two standard descriptions \(\vec{\dr} = \big((\rr_n)_{n\in\NN},\rkap\big) \in \R_c\) and \(\cev{\dr} = \big((\lr_n)_{n\in\NN},\lkap\big) \in \R_c\) of \(x_*\) that satisfy 
        \begin{enumerate}
            \item 
                \(\rr_n\leq \rr_{n+1}\) for all \(n\in\NN\),
            \item 
                \(\lr_n\geq \lr_{n+1}\) for all \(n\in\NN\),
        \end{enumerate}
        we know by \emph{Lemma \ref{T1.4}} that the pair \(\big((\rr_n)_{n\in\NN},(\lr_n)_{n\in\NN}\big)\) is sufficient for a Turing machine to determine a suitable modulus of convergence, i.e., the corresponding pair \((\rkap, \lkap)\) is obsolete in respect thereof. In this sense \emph{Theorem \ref{thm:CompleteMonRepr}} is not a direct generalization of \emph{Lemma \ref{T1.4}}, since the respective Turing machine receives a quadruple \(
        \dw = \big( (\rrp_{n,m})_{n,m\in\NN}, 
                    \rkapp~, 
                    (\lrp_{n,m})_{n,m\in\NN}, 
                    \lkapp~\big) \in \W
        \) as part of its input. In particular, the pair \((\rkapp~,\lkapp~)\) is essential in the proof of \emph{Theorem \ref{thm:CompleteMonRepr}}. This significant difference emerges from the mathematical strength of the representation of \(x_*\) in terms of computable sequences of rational numbers. If the assumptions made in \emph{Lemma \ref{T1.4}} are weakened even slightly, the corresponding proof becomes invalid. This phenomenon will be addressed further in \emph{Section \ref{comments}}.
    \end{Remark}

    Consider a family of pairs \(((\rrmn)_{n\in\NN}, (\lrmn)_{n\in\NN})_{m\in\NN}\) that satisfy the requirements of \emph{Lemma \ref{T1.4}} and each of which has their own limit value \(x_*^{m}\). In information theory, this may correspond to a specific channel model that features a model parameter \(m \in \NN\). From a practical point of view, the statement of the existence of a modulus of convergence for each pair \(((\rrmn)_{n\in\NN}, (\lrmn)_{n\in\NN})_{m\in\NN}\) is only usable because this modulus of convergence exhibits a construction algorithm that is uniformly recursive in \(((\rrmn)_{n\in\NN}, (\lrmn)_{n\in\NN})_{m\in\NN}\) (in a certain sense, we want the mapping \(m \mapsto x_*^m\) to be \emph{Turing computable}). This highlights the necessity of the \emph{converse} in the sense of the information-theoretic achievability and converse approach. Following our previous considerations, it may not be surprising that without the converse, it is generally not possible to determine the modulus of convergence of a monotonically non-decreasing sequence in a recursive way.

    \begin{Theorem}\label{T1.5}
        There exists a family \((\rrmn)_{n,m\in\NN}\) of computable sequences of rational numbers that simultaneously satisfies the following:
        \begin{itemize}
            \item 
                The family \((\rrmn)_{n,m\in\NN}\) is recursive, i.e, there exists a computable double sequence \((\rrp_{m,n})_{n,m\in\NN}\) of rational numbers that satisfies \(\rrp_{m,n} = \rrmn\) for all \(m,n\in\NN\).
            \item 
                For all \(m\in\NN\), the sequence \((\rrmn)_{n\in\NN}\) is non-negative as well as monotonically non-decreasing in \(n\) and there exists \(x_*^{m} \in \RR_c\) such that \(\lim_{n\to\infty} \rrmn = x_*^{m}\) holds true.
            \item 
                There does \emph{not} exist a Turing machine \(TM : \NN \times \NN \rightarrow \NN\) such that \(|x_*^{m} - \rrmn| < \sfrac{1}{2^M}\) holds true for all \(n,m,M\in\NN\) that satisfy \(n \geq TM(m, M)\).
        \end{itemize}
    \end{Theorem}\begin{proof} 
        We prove the statement by contradiction. Let \(A\subset \NN\) be a recursively enumerable, non-recursive set and suppose \(TM\) does exist. Since \(A\) is recursively enumerable, there exists a recursive bijection \(f_A:\NN \rightarrow A\). For all \(n\in\NN\), define 
        \begin{align}
            A_n :           &=  \big\{f_A(j) : 1 \leq j \leq n\big\} \\
            \rrmn :    &=  \left\{\begin{array}{ll} 1 & m\in A_n\\ 0& \text{otherwise} \end{array} \right. .
        \end{align}
        Then for all \(n,m\in\NN\), we have \(A_n \subset A_{n+1}\) and thus \(0 \leq \rrmn \leq \rrmnn\). Furthermore, \((\rrmn)_{n,m\in\NN}\) is a recursive family of computable sequences of rational numbers that satisfies 
        \begin{align}
            \lim_{n\to\infty} \rrmn = \mathds{1}_A(m)  
        \end{align}
        for all \(m\in\NN\). Set \(\mathds{1}_A(m) =: x_*^{m}\) for all \(m\in\NN\). By construction, we have
        \begin{align}
            \big|x_*^{m} - \rrmn\big| = \big|\mathds{1}_A(m)  - \rrmn\big|\in \{0,1\}
        \end{align}
        for all \(n,m\in\NN\). On the other hand, by definition of \(TM\), we have 
        \begin{align}
            \left|x_*^{m} - \rrmx{TM(m,1)}\right| =
            \left|\mathds{1}_A(m) -  \rrmx{TM(m,1)}\right| < \frac{1}{2}
        \end{align}
        for all \(n,m\in\NN\). Thus, for all \(m\in\NN\), we obtain
        \begin{align}
            \mathds{1}_A(m) =  \left\{\begin{array}{ll} 1 & \text{if}~ \rrmx{TM(m,1)} = 1 \\ 0& \text{otherwise} \end{array} \right. .
        \end{align}
        Consequently, \(\mathds{1}_A(m)\) is a recursive function, which contradicts the assumption of \(A\) being non-recursive.
    \qed\end{proof}

    In \emph{Section \ref{effect}} we will discuss whether Fekete's lemma, in analogy to \emph{Theorem \ref{thm:CompleteMonRepr}}, allows for the recursive construction of a suitable modulus of convergence for the sequence under consideration. \emph{Theorem \ref{T1.5}} will be essential in providing a negative answer to this question.

\section{Fekete's Lemma in the Context of the Arithmetical Hierarchy of Real Numbers}       %
\label{const}                                                                               %
 
    As indicated in \emph{Section \ref{sec:Intr}}, Fekete's lemma may prevent the general necessity of the information-theoretic converse in a certain sense: given a superadditive sequence \((\ra_n)_{n\in\NN}\) of real numbers, it provides the existence of the limit value \(\lim_{n\to\infty} \ra_n\cdot n^{-1}\). On the other hand, as we have seen in \emph{Section \ref{converge}}, the existence of a suitable converse does not only supply the existence of the limit value, but also yields a recursive construction for the modulus of convergence. As we will see in the following, this does not hold true for Fekete's lemma. In fact, the corresponding limit value may not even be a computable number. On a metamathematical level, this result implies that there cannot exist a constructive proof for Fekete's lemma, since, by definition, that would allow for an effective construction of the associated limit value.

    Throughout this section, we place Fekete's lemma in the context of the arithmetical hierarchy of real numbers. That is, we define the sets \(\Lb\) and \(\Lp\), which consist of real numbers that can be characterized by sub- and superadditive sequences and show their equivalence to the sets \(\Pi_1\) and \(\Sigma_1\). The general ineffectiveness of convergence related to Fekete's lemma then follows as a corollary.

    The following lemma is a basic result by \emph{\citeauthor{ZhWe01}}, characterizing the sets \(\Sigma_1\) and \(\Pi_1\) through the suprema and infima of computable sequences of rational numbers.

    \begin{Lemma}[\citeauthor{ZhWe01} \cite{ZhWe01}]\label{lem:WeihrauchSUP}\mbox{}
        \begin{itemize}
            \item[1.] 
                A real number \(x_*\) satisfies \(x_* \in \Sigma_1\) if and only if there exists a computable sequence \((\rr_n)_{n\in\NN}\) of rational numbers such that \(\sup_{n\in\NN} \rr_n = x_*\) holds true.
            \item[2.] 
                A real number \(x_*\) satisfies \(x_* \in \Pi_1\) if and only if there exists a computable sequence \((\lr_n)_{n\in\NN}\) of rational numbers such that \(\inf_{n\in\NN} \lr_n = x_*\) holds true.
        \end{itemize}
    \end{Lemma}

    Note that we do not require the sequences \((\rr_n)_{n\in\NN}\) and \((\lr_n)_{n\in\NN}\) to be monotonic in this context.

    Subsequently, we will make use of a generalized version of \emph{Lemma \ref{lem:WeihrauchSUP}}: the established characterization of \(\Sigma_1\) and \(\Pi_1\) extends to the suprema and infima of computable sequences of computable numbers, which we will prove in the following.

    \begin{Lemma}\label{lem:SigmaCSCN}\mbox{}
        \begin{itemize}
            \item[1.] 
                A real number \(x_*\) satisfies \(x_* \in \Sigma_1\) if and only if there exists a computable sequence \((\rz_n)_{n\in\NN}\) of computable numbers such that \(\sup_{n\in\NN} \rz_n = x_*\) holds true.
            \item[2.] 
                A real number \(x_*\) satisfies \(x_* \in \Pi_1\) if and only if there exists a computable sequence \((\lz_n)_{n\in\NN}\) of computable numbers such that \(\inf_{n\in\NN} \lz_n = x_*\) holds true. 
        \end{itemize}
    \end{Lemma}\begin{proof}\mbox{}
        \begin{itemize}
            \item[1.] 
                As every computable sequence of rational numbers is also a computable sequence of computable numbers, the implication immediately follows from \emph{Lemma \ref{lem:WeihrauchSUP}}.
            
                For the converse, consider a computable sequence \((\rz_n)_{n\in\NN}\) of computable numbers that satisfies \(\sup_{n\in\NN}\rz_n = x_*\) and a standard description \[\big((\rr_{n,m})_{n,m\in\NN},\rkap\big)\] thereof. Furthermore, consider the inverse of the \emph{Cantor pairing function} \(n \mapsto (\pi_1(n), \pi_2(n)) \in \NN\times\NN\). The Cantor pairing function \(\pi: \NN\times\NN \rightarrow \NN\) is a total and bijective recursive function, as is its inverse. Define the computable sequence \((\rrp_n)_{n\in\NN}\) of rational numbers by setting
                \begin{align}
                    \rrp_n := \rr_{\pi_1(n), \rkap(\pi_1(n),\pi_2(n))} - \frac{1}{2^{\pi_2(n)}}
                \end{align}
                for all \(n\in\NN\). Then, we have
                \begin{align}
                    \rrp_n \leq \rz_{\pi_1(n)} \leq \sup_{m\in\NN} \rz_m = x_*
                \end{align}
                for all \(n\in\NN\), i.e., \(\sup_{n\in\NN} \rrp_n \leq x_*\). It remains to show that \(\sup_{n\in\NN} \rrp_n = x_*\) holds true. Let \( 0 < \epsilon < 1\) be arbitrary. Since we have \(\sup_{n\in\NN} \rz_n = x_*\) as well as
                \begin{align}
                    \lim_{m\to\infty} \left(\rr_{n,\rkap(n,m)} - \frac{1}{2^m}\right) = \nu_n 
                \end{align}
                for all \(n\in\NN\), there exist \(l,k\in\NN\) such that
                \begin{align}
                    x_* - \epsilon \leq \rr_{l,\rkap(l,k)} - \frac{1}{2^{k}} \leq x_* 
                \end{align}
                holds true. Furthermore, there exists \(n\in\NN\) such that \((\pi_1(n), \pi_2(n)) = (l,k)\) is satisfied. Therefore, there exists \(n\in\NN\) such that 
                \begin{align}
                    x_* - \epsilon \leq \rrp_{n} \leq x_*    
                \end{align}
                holds true. Since \(\epsilon\) was chosen arbitrarily, we have \(\sup_{n\in\NN} \rrp_n = x_*\). Hence, \(x_*\in\Sigma_1\) follows by \emph{Lemma \ref{lem:WeihrauchSUP}}.
            \item[2.] 
                The proof of the second claim follows along the same line of reasoning as the proof of the first claim. Thus, we restrict ourselves to a brief summary of the steps. 
            
                Again, as every computable sequence of rational numbers is also a computable sequence of computable numbers, the implication immediately follows from \emph{Lemma \ref{lem:WeihrauchSUP}}.
                
                For the converse, consider a standard description \(\big((\lr_{n,m})_{n,m\in\NN},\lkap\big)\) of a suitable computable sequence of computable numbers and define
                \begin{align}
                    \lrp_n := \lr_{\pi_1(n), \sublkap(\pi_1(n),\pi_2(n))} + \frac{1}{2^{\pi_2(n)}}
                \end{align}
                for all \(n\in\NN\). The computable sequence \((\lrp_n)_{n\in\NN}\) of rational numbers then satisfies \(\inf_{n\in\NN} \lrp_n = x_*\). Hence, \(x_* \in \Pi_1\) follows by \emph{Lemma \ref{lem:WeihrauchSUP}}.
        \end{itemize}
    \qed\end{proof}

    As indicated in \emph{Section \ref{sec:Intr}}, the arithmetical hierarchy by \emph{\citeauthor{ZhWe01}} closely relates to the Kleene–Mostowski hierachy of sets of natural numbers. In view of the denomination "Kleene–Mostowski hierachy," we will indentify the characterization of real numbers by means of the suprema and infima of computable sequences of computable numbers by the term "(\(n\)-th order) Zheng-Weihrauch representation."

    We refer to a bounded, computable sequence \((\rz_n)_{n\in\NN}\) of computable numbers as a \emph{lower first-order Zheng-Weihrauch (ZW) representation} of the real number \(x_* = \sup_{n\in\NN} \rz_n\). The set of standard descriptions  \(\vec{\dz} = \big((\rr_{n,m})_{n,m\in\NN},\rkap\big)\) of lower first-order ZW representations of real numbers is denoted by \(\vec{\Z}\). 

    Likewise, we refer to a bounded, computable sequence \((\lz_n)_{n\in\NN}\) of computable numbers as an \emph{upper first-order Zheng-Weihrauch (ZW) representation} of the real number \(x_* = \inf_{n\in\NN} \lz_n\). The set of standard descriptions \(\cev{\dz} = \big((\lr_{n,m})_{n,m\in\NN},\lkap\big)\) of upper first-order ZW representations of real numbers is denoted by \(\cev{\Z}\).

    Ultimately, our goal is to introduce a characterization of (semi-)computable numbers through super- and subadditive sequences, which subsequently allows us to relate Fekete's lemma to the arithmetical hierarchy of real numbers. Thus, in analogy to \(\Pi_1\) and \(\Sigma_1\), define the sets \(\Lp \subseteq \RR\) and \(\Lb \subseteq \RR\) as follows:
    \begin{itemize}
        \item   $\Lb     := \big\{ x_* \in \RR ~:~$ There exists a \emph{subadditive}, computable sequence \((\la_n)_{n\in\NN}\) of computable numbers that satisfies
            $\lim_{n\to\infty} \la_n\cdot n^{-1} = x_*.\big\}$, 
    
        \item $\Lp     := \big\{ x_* \in \RR ~:~$   There exists a \emph{superadditive}, computable sequence \((\ra_n)_{n\in\NN}\) of computable numbers that satisfies $\lim_{n\to\infty} \ra_n\cdot n^{-1} = x_*.\big\}$. 
    \end{itemize}

    Similar to the notion of first-order ZW representations, we refer to a superadditive, computable sequence \((\ra_n)_{n\in\NN}\) of computable numbers that satisfies \(\lim_{n\to\infty} \ra_n \cdot n^{-1} = x_*\) for some real number \(x_*\) as a \emph{superadditive representation} of \(x_*\). The set of standard descriptions \(\vec{\da} = \big((\rr_{n,m})_{n,m\in\NN},\rkap\big)\) of superadditive representations of real numbers is denoted by \(\vec{\A}\).

    Likewise, we refer to a subadditive, computable sequence \((\la_n)_{n\in\NN}\) of computable numbers that satisfies \(\lim_{n\to\infty} \la_n \cdot n^{-1} = x_*\) for some real number \(x_*\) as a \emph{subadditive representation} of \(x_*\). The set of standard descriptions \(\cev{\da} = \big((\lr_{n,m})_{n,m\in\NN},\lkap\big)\) of subadditive representations of real numbers is denoted by \(\cev{\A}\).

    Given the definition of the sets \(\Lp\) and \(\Lb\) as well as \emph{Lemma \ref{lem:SigmaCSCN}}, we can immediately prove the following relation to the sets \(\Pi_1\) and \(\Sigma_1\):

    \begin{Lemma}\label{lem:LsupLsubSUBSET}
        For all real numbers \(x_*\), the following holds true:
        \begin{itemize}
            \item[1.] 
                If there exists a superadditive representation for \(x_*\), then there exists a lower first-order ZW representation for \(x_*\). Consequently, we have \(\Lp \subseteq \Sigma_1\).
            \item[2.] 
                If there exists a subadditive representation for \(x_*\), then there exists an upper first-order ZW representation for \(x_*\). Consequently, we have \(\Lb \subseteq \Pi_1\).
        \end{itemize}
    \end{Lemma}\begin{proof}\mbox{}
        \begin{itemize}
            \item[1.] 
                Consider a superadditive representation \((\ra_n)_{n\in\NN}\) for \(x_*\) and set \(\rz_n := \ra_n \cdot n^{-1}\) for all \(n\in\NN\). Following \emph{Lemma \ref{lem:fekete}}, we conclude that 
                \begin{align}
                    x_* = \lim_{n\to\infty} \ra_n \cdot n^{-1} = \sup_{n\in\NN} \rz_n
                \end{align} 
                holds true. Thus, \((\rz_n)_{n\in\NN}\) is a first-order ZW representation of \(x_*\), and, by \emph{Lemma \ref{lem:SigmaCSCN}}, we have \(x_* \in \Sigma_1\).
            \item[2.] 
                Consider a subadditive representation \((\la_n)_{n\in\NN}\) for \(x_*\) and set \(\lz_n := \la_n \cdot n^{-1}\) for all \(n\in\NN\). Following \emph{Lemma \ref{lem:fekete}}, we conclude that 
                \begin{align}
                    x_* = \lim_{n\to\infty} \la_n \cdot n^{-1} = \inf_{n\in\NN} \lz_n
                \end{align} 
                holds true. Thus, \((\lz_n)_{n\in\NN}\) is a first-order ZW representation of \(x_*\), and, by \emph{Lemma \ref{lem:SigmaCSCN}}, we have \(x_* \in \Pi_1\).
        \end{itemize}
    \qed\end{proof}

    Observe that the proof of \emph{Lemma \ref{lem:LsupLsubSUBSET}} is constructive. Hence, we can transform a super(sub)additive representation of a number \(x_* \in\RR\) effectively into a lower (upper) first-order ZW representation of \(x_*\). As we will prove in the following, the converse is true as well.

    \begin{Lemma}\label{lem:SigmaSUBSET}
        There exists a Turing machine \(TM : \vec{\Z} \rightarrow \vec{\A}\) that satisfies the following:
        \begin{itemize}
            \item 
                If the input of \(TM\) is the standard description of a lower first-order ZW representation of a number \(x_*\), then the output of \(TM\) is the standard description of a superadditive representation of \(x_*\).
        \end{itemize}
        Consequently, we have \(\Sigma_1 \subseteq \Lp\).
    \end{Lemma}\begin{proof}
        Denote again by \(n \mapsto (\pi_1(n), \pi_2(n)) \in \NN\times\NN\) the inverse of the Cantor pairing function. Let \(\big((\rr_{n,m})_{n,m\in\NN},\rkap\big)\) be the standard description of a lower first-order ZW representation \((\rz_n)_{n\in\NN}\) of \(x_*\). We define
        \begin{equation}
            \rrp_n := \max  \bigg\{\rr_{\pi_1(m), \rkap(\pi_1(m),\pi_2(m))} - \frac{1}{2^{\pi_2(m)}}~:
                             1 \leq m \leq n\bigg\}
        \end{equation}
        for all \(n\in\NN\). Following the line of reasoning presented in the proof of \emph{Lemma \ref{lem:SigmaCSCN}}, we conclude that \((\rrp_n)_{n\in\NN}\) is a monotonically non-decreasing, computable sequence of rational numbers that satisfies \(\lim_{n\to\infty} \rrp_n = \sup_{n\in\NN} \rrp_n = \sup_{n\in\NN} \rz_n = x_*\) and that the mapping \(\big((\rr_{n,m})_{n,m\in\NN},\rkap\big) \mapsto (\rrp_n)_{n\in\NN}\) is recursive. Observe that the sequence \((n \cdot \rrp_n)_{n\in\NN}\) satisfies superadditivity, since
        \begin{align}
                (n+m) \cdot \rrpx{n+m}  &= n \cdot \rrpx{n+m} + m \cdot \rrpx{n+m}\nonumber\\ 
                                        &\geq n \cdot \rrp_{n} + m \cdot \rrp_{m} 
        \end{align}
        holds true for all \(n,m\in\NN\), following the monotonicity of \((\rrp_n)_{n\in\NN}\).
        
        Now consider the computable double sequence \((\rrppx{n,m})_{n,m\in\NN}\) of rational numbers defined via
        \begin{align}
            \rrppx{n,m} := n \cdot \rrp_n  
        \end{align}
        for all \(n,m\in\NN\), as well as the recursive function \(\rkapp~: \NN\rightarrow \NN,~ M \mapsto 1\). Clearly, \((\rrppx{n,m})_{n,m\in\NN}\) satisfies \(\lim_{m\to\infty} \rrppx{n,m} = n \cdot \rrp_n \) for all \(n\in\NN\), as well as 
        \begin{align}
            \left|n \cdot \rrp_n - \rrppx{n,m} \right| = 0 < \frac{1}{2} = \frac{1}{2^{\rkapp~(M)}}
        \end{align}
        for all \(n,m,M\in\NN\) that satisfy \(m \geq \rkapp(M)\). Thus, the pair \(\big((\rrppx{n,m})_{n,m\in\NN},\rkapp\hspace{3.5pt}\big)\) is the standard description of a superadditive representation of the number \(x_*\) and the mapping \(\big((\rr_{n,m})_{n,m\in\NN},\rkap\big) \mapsto \big((\rrppx{n,m})_{n,m\in\NN},\rkapp\hspace{3.5pt}\big)\) is recursive.
    \qed\end{proof}

    \begin{Lemma}\label{lem:PiSUBSET}
        There exists a Turing machine \(TM : \cev{\Z} \rightarrow \cev{\A}\) that satisfies the following:
        \begin{itemize}
            \item 
                If the input of \(TM\) is the standard description of an upper first-order ZW representation of a number \(x_*\), then the output of \(TM\) is the standard description of a subadditive representation of \(x_*\).
        \end{itemize}
        Consequently, we have \(\Pi_1 \subseteq \Lb\).
    \end{Lemma}\begin{proof} 
        The proof of the claim follows along the same line of reasoning as the proof of \emph{Lemma \ref{lem:SigmaSUBSET}}. Thus, we will restrict ourselves to a brief summary of the steps.
    
        Again, consider a standard description\linebreak \(\big((\lr_{n,m})_{n,m\in\NN},\lkap\big)\) of an upper first-order ZW representation of the number \(x_*\in\RR\). For all \(n,m\in\NN\), we define
        \begin{equation}
            \lrppx{n,m} := n\cdot \min  \bigg\{\lr_{\pi_1(l), \sublkap(\pi_1(l),\pi_2(l))} + \frac{1}{2^{\pi_2(l)}}~: 1 \leq l \leq n\bigg\}.
        \end{equation}
        Furthermore, set \(\lkapp~ : \NN\rightarrow \NN,~ M \mapsto 1\). Then the mapping \(\big((\lr_{n,m})_{n,m\in\NN},\lkap\big) \mapsto \big((\lrppx{n,m})_{n,m\in\NN},\lkapp\hspace{3.5pt}\big)\) is recursive and the pair \(\big((\lrppx{n,m})_{n,m\in\NN},\lkapp\hspace{3.5pt}\big)\) is a standard description of a subadditive representation of \(x_*\).
    \qed\end{proof}

    We conclude that we are likewise able to transform any lower (upper) first-order ZW representation effectively into a super(sub)additive representation of the same number, leading up to the main result of this section: 

    \begin{Theorem}\label{thm:LsubLsupRc}
        For the sets \(\Sigma_1\), \(\Lp\), \(\Pi_1\), \(\Lb\) and \(\RR_c\), the following equalities hold true:
        \begin{align}
                \Sigma_1    &= \Lp, \label{eq:thm:LsubLsupRc::I}\\
                \Pi_1       &= \Lb, \label{eq:thm:LsubLsupRc::II}\\
                \RR_c       &= \Lp \cap \Lb. \label{eq:thm:LsubLsupRc::III}
        \end{align}
    \end{Theorem}\begin{proof}
        The logical conjunction of \emph{Lemma \ref{lem:LsupLsubSUBSET}} and \emph{Lemma \ref{lem:SigmaSUBSET}} yields \eqref{eq:thm:LsubLsupRc::I}, while \eqref{eq:thm:LsubLsupRc::II} is the logical conjunction of \emph{Lemma \ref{lem:LsupLsubSUBSET}} and \emph{Lemma \ref{lem:PiSUBSET}}. As indicated in the beginning of \emph{Section \ref{converge}}, we have \(\RR_c = \Pi_1 \cap \Sigma_1\) \cite{ZhWe01}, which subsequently implies \eqref{eq:thm:LsubLsupRc::III}.
    \qed\end{proof}

    \begin{Corollary}\label{C1.1}
        There exists a superadditive, computable sequence \((\ra_n)_{n\in\NN}\) of computable numbers as well as a real number \(x_* \in \RR\setminus\RR_c\) such that \(\lim_{n\to\infty} \ra_n\cdot n^{-1} = c_*\) holds true.
    \end{Corollary}\begin{proof}
        As indicated in the beginning of \emph{Section \ref{converge}}, we have \(\Sigma_1\setminus\RR_c \neq \emptyset\) \cite{ZhWe01}. The assertion then follows from \emph{Theorem \ref{thm:LsubLsupRc}} and the definition of the set \(\Lp\).
    \qed\end{proof}
 
    \begin{Remark}
        According to Specker \cite{Spe49}, we can find a non-negative, monotonically non-decreasing, computable sequence \((\rr_n)_{n\in\NN}\) of rational numbers such that the corresponding limit value \(x_* = \lim_{n\to\infty} \rr_n\) is \emph{not} a computable number. The sequence \((\rrp_n)_{n\in\NN}\), defined via \(\rrp_n := n\cdot \rr_n\) for all \(n\in\NN\) then yields an alternative, direct proof of \emph{Corollary \ref{C1.1}}.
    \end{Remark}


    \begin{Remark}\label{RfooI}
        \emph{Corollary \ref{C1.1}} highlights the boundaries of the state-of-the-art methods in information theory. For the number \(x_*\), there cannot exist a monotonically non-increasing, computable sequence of computable numbers with limit value \(x_*\). This is due to the fact that \(x_*\) is the limit value of a monotonically non-decreasing computable sequence of rational numbers on the one hand, but an uncomputable number on the other. Thus, the approach of characterizing \(x_*\) by deriving a suitable converse (in the sense of information theory) is not possible in this case. Surprisingly, this phenomenon is \emph{not} exclusive to the information theoretic converse, as was shown in \cite{BoScPo20:2}: 
        
        Denote by \(\mathcal{CH}_{c}(\{0,1\},\{0,1\})\) the set of computable stochastic \(2 \times 2\) matrices.  There exists a \emph{compound channel} \(\V_* := \{V_s\}_{s\in\NN}\) with \(V_s \in \mathcal{CH}_{c}(\{0,1\},\{0,1\})\) for all \(s\in\NN\), such that the mapping \(K \mapsto \V_{K} := \{V_s\}_{s\in\{1,2,\ldots,K\}}\) is recursive and the capacity \(C(\V_*)\) exists, but is \emph{not} a computable number. On the other hand, the capacity \(C(\V_K)\) exists and is computable for all \(K\in\NN\), since \(C(\V_K)\) only contains a finite number of components. It was furthermore shown that \(\lim_{K\to\infty} C(\V_K) = C(\V_*)\) holds true, with \((C(\V_K))_{K\in\NN}\) being a monotonically non-increasing sequence. Thus, there cannot exist a computable sequence of computable numbers that converges towards \(C(\V_*)\) from below. Again, denote by \(M(n)\) the maximum number of messages transmittable by \(n\) uses of the channel. As indicated in the introduction, the existence of an optimal code for each blocklength \(n\) is straightforward, since the number of feasible codes is finite. Furthermore, for the channel \(\V_*\), we have \(\sfrac{1}{n}\cdot\log_2 M(n) \leq C(\V_*)\) for all \(n\in\NN\) (with respect to an error constraint of \(0 < \epsilon < \sfrac{1}{2}\)) and, by the existence of \(C(\V_*)\), \(\lim_{n\to\infty} \sfrac{1}{n}\cdot\log_2 M(n) = C(\V_*)\). Thus, the mapping \(n\mapsto M(n)\) is nonrecursive in this context. Consequently, the common information theoretic approach of finding optimal codes for an increasing number of channel uses \(n\) is not possible for the channel \(\V_*\).
    \end{Remark}

    \begin{Remark}\label{non}
        As indicated in the \emph{Section \ref{sec:Intr}}, Fekete's lemma provides an multi-letter description for the capacities associated with a number of channel models \cite{ABBN13,BBS13, WNB16}. In his Shannon lecture in 2008, Rudolf Ahlswede pointed out that a lot of research concerning single-letter descriptions of channel capacities has been conducted during the past 60 years. For many capacities of practical interest, however, it remains unclear whether a single-letter description does exist. In view of this problem, Ahlswede asked whether it would instead be possible to prove the effective convergence of multi-letter formulas towards the corresponding capacity. According to \emph{Definition \ref{defeff}}, it would then be possible to compute the numerical value of the capacity up to any desired accuracy. For many practical applications, this would be sufficient for the near future. As implicated by \emph{Corollary \ref{C1.1}}, Ahlswede's question cannot be answered positively on the basis of Fekete's lemma, since the corresponding limit value does not necessarily have to be a computable number. Furthermore, we can observe that the non-computability of the limit value occurs for the de Bruijn-Erdös condition as well (\cite{BE51,BE52}).
    \end{Remark}

\section{An Effective Version of Fekete's Lemma}        %
\label{effect}                                          %

    In this section, we prove an \emph{effective} version of Fekete's lemma for superadditive sequences that satisfy non-negativity. 

    Given a computable sequence \((r_n)_{n\in\NN}\) of rational numbers that converges towards a real number \(x_*\) with recursive modulus of convergence \(\kappa : \NN \rightarrow \NN\), we immediately have \(x_* \in \RR_c\), due to the definition of computable numbers. As indicated before, the converse is not true in general. Given a computable sequence \((r_n)_{n\in\NN}\) of rational numbers that converges towards a computable number \(x_*\), there may not exist a recursive modulus of convergence. A general description of this phenomenon may be found in \cite{Vy98}. In the following, we will show that the equivalence holds true for non-negative sequences in the context of Fekete's lemma. The majority of sequences considered in discrete mathematics and information theory do, in fact, satisfy non-negativity. For didactic reasons, we first establish the equivalence for computable sequences of rational numbers and present the equivalence for computable sequences of computable numbers subsequently.  

    \begin{Lemma}\label{T1.2}
        Let \((\rr_n)_{n\in\NN}\) be a superadditive, computable sequence of rational numbers that satisfies \(0 \leq \rr_n\) for all \(n\in\NN\). Furthermore, let \(x_*\) be a real number such that \(\lim_{n\to\infty} \rr_n\cdot n^{-1} = x_*\) holds true. The sequence \((\rr_n\cdot n^{-1})_{n\in\NN}\) converges \emph{effectively} towards \(x_*\) if and only if \(x_*\) is a computable number.
    \end{Lemma}\begin{proof}
        If the sequence \((\rr_n\cdot n^{-1})_{n\in\NN}\) converges effectively towards \(x_*\), then \(x_*\) is computable, as is clear from the definition of computable numbers. 
    
        For the converse, consider a monotonically non-increasing sequence \((\lr_n)_{n\in\NN}\) of rational numbers that satisfies \(\lim_{n\to\infty} \lr_n = x_*\). The existence of such a sequence is ensured by \(x_*\) being a computable number. Furthermore, by \emph{Lemma \ref{lem:fekete}}, we have \( \rr_n\cdot n^{-1} \leq x_*\) for all \(n\in\NN\). For \(M\in\NN\), we define
        \begin{equation}\label{eq:nO}
            n_0 := \min \bigg\{n \in \NN ~:~    \frac {\lr_n}{n} < \frac{1}{2^{M+1}} ~\wedge\lr_n-\frac{1}{2^{M+1}} < \frac{\rr_n}{n}\bigg\}.
        \end{equation}
        Since both \((\rr_n \cdot n^{-1})_{n\in\NN}\) and \((\lr_n)_{n\in\NN}\) converge to the same number, we know that \(n_0\) must exist. Furthermore, both conditions in \eqref{eq:nO} can be decided algorithmically, since they exclusively involve comparisions on rational numbers. Thus, given the sequences \((\rr_n)_{n\in\NN}\) and  \((\lr_n)_{n\in\NN}\), the mapping \(M \mapsto n_0\) is recursive. For all \(n\in\NN\) that satisfy \(n\geq n_0\), there exist two unique numbers \(q,s\in \NN\) with \(0\leq s\leq n_0-1\), such that \(n=qn_0+s\) holds true. We have
        \begin{align}
            \rr_{qn_0+s}    &\geq \rr_{qn_0}+\rr_{s} \nonumber\\
                            &\geq \rr_{n_0} + \rr_{(q-1)n_0}+\rr_{s} \nonumber\\
                            &\geq \rr_{n_0} + \rr_{n_0} + \rr_{(q-2)n_0}+\rr_{s} \nonumber\\
                            &\geq \quad \vdots \nonumber\\
                            &= q\rr_{n_0}+\rr_{s}.
        \end{align}
        Thus, \(q\rr_{n_0}+\rr_{s} \leq \rr_{qn_0+s} = \rr_n\) is satisfied, and consequently,
        \begin{align}
            \frac {\rr_n}{n}    &\geq \frac {q\rr_{n_0}}{n}+\frac {\rr_s}{n}\nonumber\\
                                &= \frac {qn_0}{n} \frac {\rr_{n_0}}{n_0} +\frac {\rr_s}{n}\nonumber\\
                                &> \frac {qn_0}{n} \lr_{n_0} - \frac {qn_0}{n}\frac {1}{2^{M+1}} +\frac {\rr_s}{n}\nonumber\\
                                &> \frac {qn_0}{n} \lr_{n_0} - \frac {1}{2^{M+1}}
        \end{align}
        holds true. Now consider \(n\in\NN\) satisfying \(n \geq n_0^2\). Since any \(n\in\NN\) that satisfies \(n \geq n_0^2\) certainly satisfies \(n \geq n_0\) as well, there again exist unique numbers \(q,s \in \NN\) with \(0\leq s\leq n_0-1\), such that \(n=qn_0+s\) holds true. Thus, we have
        \begin{align}
            0       &\leq \lr_n-\frac {\rr_n}{n}\nonumber\\
                    &\leq \lr_{n_0}-\frac {\rr_n}{n}\nonumber\\
                    &\leq \lr_{n_0}- \frac {qn_0}{n}\lr_{n_0}+\frac {1}{2^{M+1}}\nonumber\\
                    &\leq \left( 1-\frac{qn_0}{n}\right)\lr_{n_0}+\frac 1{2^{M+1}}\nonumber\\
                    &\leq \frac{\lr_{n_0}}{n_0} +\frac 1{2^{M+1}}\label{eq:FooEQ}\\
                    &< \frac 1{2^{M+1}}+\frac 1{2^{M+1}}\nonumber\\ 
                    &= \frac 1{2^{M}},
        \end{align}
        where \eqref{eq:FooEQ} holds true because
        \begin{align}
            1-\frac{qn_0}{n}=\frac {n-qn_0}{n}\leq  \frac {n_0}{n}\leq \frac {n_0}{n_0^2} = \frac{1}{n_0}
        \end{align} 
        is satisfied. Define the mapping \(\kappa:\NN \rightarrow \NN, M \mapsto n_0^2\), which, given the sequences \((\rr_n)_{n\in\NN}\) and \((\lr_n)_{n\in\NN}\), is a recursive function. Then, we have 
        \begin{align}
            \left|x_*-\frac {\rr_n}{n}\right|=x_*-\frac{\rr_n}{n}\leq \lr_n-\frac {\rr_n}n\leq \frac 1{2^M}
        \end{align}
        for all \(n,M\in\NN\) that satisfy \(n \geq \kappa(M)\). Thus, \(\kappa\) yields the required modulus of convergence for the sequence \((\rr_n\cdot n^{-1})_{n\in\NN}\).
    \qed\end{proof}

    \begin{Theorem}\label{T1.3}
        Let \((\ra_n)_{n\in\NN}\) be a superadditive representation of the number \(x_* \in \RR\) that satisfies \( 0 \leq \ra_n \) for all \(n\in\NN\). The sequence \((\ra_n\cdot n^{-1})_{n\in\NN}\) converges \emph{effectively} towards \(x_*\) if and only if \(x_*\) is a computable number.
    \end{Theorem}\begin{proof}
        In the following, let \(\big((\rr_{n,m})_{n,m\in\NN},\rkap\big)\) be a standard description of the sequence \((\ra_n)_{n\in\NN}\).
        
        If \((\ra_n)_{n\in\NN}\) converges effectively towards \(x_*\) with recursive modulus of convergence \(\kappa:\NN \rightarrow \NN\), it is straightforward to construct a computable sequence of rational numbers from \(\big((\rr_{n,m})_{n,m\in\NN},\rkap\big)\) and \(\kappa\) that converges effectively towards \(x_*\). Then, \(x_*\) is a computable number.
        
        For the converse, we will show the existence of a non-negative, superadditive, computable sequence \((\rrppx{n})_{n\in\NN}\) of rational numbers that satisfies \(\rrppx{n} \leq \ra_n\) for all \(n\in\NN\) as well as \(\lim_{n\to\infty} \rrppx{n}\cdot n^{-1} = x_*\). The theorem then follows by application of \emph{Lemma \ref{T1.2}}. Define the sequence \((\rrp_{n})_{n\in\NN}\) via
        \begin{align}
            \rrp_{n} := \max\left\{0,~ \rr_{n,\rkap(n,n)} - \frac{1}{2^n}\right\}
        \end{align}
        for all \(n\in\NN\). Then, \((\rrp_{n})_{n\in\NN}\) is a computable sequence of rational numbers that satisfies \(\rrp_{n} \leq \ra_n\) for all \(n\in\NN\) as well as \(\lim_{n \to\infty} \rrp_{n} \cdot n^{-1} = x_*\). We now define the sequence \((\rrppx{n}\hspace{1pt})_{n\in\NN}\) in an inductive manner.
        Set \(\rrppx{1} := \rrp_1\), and for all \(n\in\NN\) that satisfy \(n \geq 2\),
        \begin{align}
            s_n         :&= \max\big\{\rrppx{l}\hspace{2pt} + \rrppx{k} : l,k\in \NN, l + k = n\big\}, \\
            \rrppx{n}   :&= \max\big\{\rrp_{n},~ s_n\big\}. 
        \end{align}
        Then, \((\rrppx{n}\hspace{1pt})_{n\in\NN}\) is a non-negative, superadditive, computable sequence of rational numbers that satisfies \(\rrp_{n} \leq \rrppx{n}\) for all \(n\in \NN\). Assume \(\rrppx{n} = \rrp_{n}\) for some \(n\in\NN\). Then, \(\rrppx{n} \leq \ra_n\) holds true.
        Assume now \(\rrppx{n} = \rrppx{l} + \rrppx{k}\) for some \(l,k,n\in\NN\) that satisfy \(n = l + k\). Then, by induction,
        \begin{align}
            \rrppx{n} = \rrppx{l} + \rrppx{k} \leq \ra_{l} + \ra_{k} \leq \ra_{l + k} = \ra_n  
        \end{align}
        holds true. Thus, the sequence \((\rrppx{n})_{n\in\NN}\) satisfies \(\rrp_{n} \leq \rrppx{n} \leq \ra_n\) for all \(n\in\NN\). Therefore, we have \(\rrp_{n}\cdot n^{-1} \leq \rrppx{n} \cdot n^{-1} \leq \ra_n\cdot n^{-1} \leq x_*\) for all \(n\in\NN\), as well as \(\lim_{n\to\infty} \rrppx{n}\cdot n^{-1} = x_*\). Any modulus of convergence for the sequence \((\rrppx{n}\cdot n^{-1})_{n\in\NN}\) is thus a modulus of convergence for the sequence \((\ra_n\cdot n^{-1})_{n\in\NN}\) as well. The theorem then follows from \emph{Lemma \ref{T1.2}}.
    \qed\end{proof}

    \begin{Remark}\label{rem:ProofsPartwiseConstructive}
        Observe that apart from the necessity of finding a suitable sequence \((\lr_n)_{n\in\NN}\), both the proof of \emph{Lemma \ref{T1.2}} and the proof of \emph{Theorem \ref{T1.3}} are constructive. Thus, given a suitable pair \(\big((\lr_n)_{n\in\NN},(\ra_n)_{n\in\NN}\big)\), we can find a modulus of convergence \(\kappa\) for the sequence \((\ra_n\cdot n^{-1})_{n\in\NN}\) by means of a Turing machine.   
    \end{Remark}

    \begin{Remark}
        In the above sense, we have a positive answer to Ahlswede's question (cf. \emph{Remark} \ref{non}) for computable capacities that exhibit a non-negative, superadditive representation.
    \end{Remark}

    We have discussed the necessity of algorithmic constructions for moduli of convergence to a great extent in the previous part of our work. While \emph{Lemma \ref{T1.2}} and \emph{Theorem \ref{T1.3}} prove the existence of a recursive modulus of convergence in the context of the effective version of Fekete's lemma, they do so in a non-constructive way. In the following, we will apply \emph{Theorem \ref{T1.5}} to show that even the effective version of Fekete's lemma is not strong enough to allow for a recursive construction of the associated modulus of convergence.

    \begin{Theorem}\label{T1.7}
        There exists a family \((\rrmn)_{n,m\in\NN}\) of computable sequences of rational numbers that simultaneously satisfies the following:
        \begin{itemize}
            \item 
                The family \((\rrmn)_{n,m\in\NN}\) is recursive, i.e, there exists a computable double sequence \((\rrp_{m,n})_{n,m\in\NN}\) of rational numbers that satisfies \(\rrp_{m,n} = \rrmn\) for all \(m,n\in\NN\).
            \item 
                For all \(m\in\NN\), the sequence \((\rrmn)_{n\in\NN}\) satisfies non-negativity  as well as superadditivity in \(n\), and there exists \(x_*^{m} \in \RR_c\) such that \(\lim_{n\to\infty} \rrmn \cdot n^{-1} = x_*^{m}\) holds true.
            \item 
                There does \emph{not} exist a Turing machine \(TM : \NN \times \NN \rightarrow \NN\) such that \(|x_*^{m} - \rrmn\cdot n^{-1}| < \sfrac{1}{2^M}\) holds true for all \(n,m,M\in\NN\) that satisfy \(n \geq TM(m, M)\).
        \end{itemize}
    \end{Theorem}\begin{proof} 
        We prove the statement by contradiction. Let \((\rrppx{n,m})_{n,m\in\NN}\) be computable double sequences of rational numbers that satisfy the requirements of \emph{Theorem \ref{T1.5}}, and suppose \(TM\) does exist. For all \(n,m\in\NN\), define
        \begin{align}
            \rrmn := n\cdot \rrppx{n,m}.
        \end{align}
        Then, \((\rrmn)_{n,m\in\NN}\) is a recursive family of computable sequences of rational numbers that satisfies non-negativity as well as superadditivity in \(n\) for all \(m\in\NN\). Furthermore, for all \(m\in\NN\), there exists \(x_*^{m}\in\RR_c\) such that
        \begin{align}
            \lim_{n\to\infty} \rrmn\cdot n^{-1} = \lim_{n\to\infty} \rrppx{n,m} = x_*^{m}
        \end{align}
        holds true. By assumption, the Turing machine \(TM\) satisfies
        \begin{align}
            \big|x_*^{m} - \rrmn\cdot n^{-1} \big| = 
            \big|x_*^{m} - \rrppx{n,m}       \big| < 2^{-M}
        \end{align}
        for all \(n,m,M\in\NN\) that satisfy \(n \geq TM(m, M)\), thereby contradicting \emph{Theorem \ref{T1.5}}.
    \qed\end{proof}

    In \emph{Section \ref{converge}}, we have introduced the notion of monotonic representations for computable numbers. By \emph{Theorem \ref{thm:CompleteMonRepr}}, we know that given a monotonic representation \(\big((\rw_n)_{n\in\NN},(\lw_n)_{n\in\NN}\big)\) for some number \(x_*\in\RR\), we can algorithmically deduce moduli of convergence for the sequences \((\rw_n)_{n\in\NN}\) and \((\lw_n)_{n\in\NN}\). On the other hand, we know by \emph{Theorem \ref{thm:LsubLsupRc}} that a number \(x_*\in\RR\) is computable if and only if it exibits both a super- and a subadditive representation. In view of this result, we conclude the section by proving a claim similar to \emph{Theorem \ref{thm:CompleteMonRepr}} for non-negative super- and subadditive representations. 

    In the following, we refer to a pair \(\big((\ra_n)_{n\in\NN},(\la_n)_{n\in\NN}\big)\) as a \emph{P/B-additive representation} of \(x_*\in\RR\) if \((\ra_n)_{n\in\NN}\) is a superadditive representation of \(x_*\) and \((\la_n)_{n\in\NN}\) is a subadditive representation of \(x_*\). Furthermore, we denote by \(\A\) the set of standard descriptions
    \begin{align}
        \da = \big((\rr_{n,m})_{n,m\in\NN}, \rkap, (\lr_{n,m})_{n,m\in\NN}, \lkap\big)
    \end{align}
    of P/B-additive representations of real numbers.

    \begin{Theorem}\label{thm:analogonFekete}
        There exists a Turing machine \(TM : A\times \NN \rightarrow \NN\) that satisfies the following:
        \begin{itemize}
            \item 
                If \(\big((\ra_n)_{n\in\NN},(\la_n)_{n\in\NN}\big)\) is a P/B-additive representation of \(x_* \in \RR\) that satisfies \(0 \leq \ra_n\) for all \(n\in\NN\) and \(\da\) is a standard description of \(\big((\ra_n)_{n\in\NN},(\la_n)_{n\in\NN}\big)\), then 
                \begin{align}
                    \big|x_*-\ra_n \cdot n^{-1}\big|<\frac 1{2^M}
                \end{align}
                holds true for all \(n\in\NN\) that satisfy \(n\geq TM(\da,M)\). 
            \end{itemize}
    \end{Theorem}\begin{proof}
        As \emph{Remark \ref{rem:ProofsPartwiseConstructive}} points out, both the proof of \emph{Lemma \ref{T1.2}} and the proof of \emph{Theorem \ref{T1.3}} are constructive, apart from the necessity of finding a monotonically non-increasing, computable sequence that converges towards \(x_*\) from above. We will thus restrict ourselves to a brief summary of the necessary steps. 
        
        Consider the standard description \(\da = \big((\rr_{n,m})_{n,m\in\NN}, \rkap, (\lr_{n,m})_{n,m\in\NN}, \lkap\big)\) of \(\big((\ra_n)_{n\in\NN},(\la_n)_{n\in\NN}\big)\) and define
        \begin{align}
            \lrp_n := \min\left\{ \lr_{m,\sublkap(m,m)} + \frac{1}{2^m} : 1 \leq m \leq n \right\}
        \end{align}
        for all \(n\in\NN\). The sequence \((\lrp_n)_{n\in\NN}\) is a monotonically non-increasing, computable sequence of rational numbers that satisfies \(\lim_{n\to\infty} \lrp_n = x_*\) and the mapping \(\big((\lr_{n,m})_{n,m\in\NN},\lkap\big) \mapsto (\lrp_n)_{n\in\NN}\) is recursive. Furthermore, define 
        \begin{align}
            \rrppx{n} : &=  \max\left\{0,~ \rr_{n,\rkap(n,n)} - \frac{1}{2^n}\right\}, \\
            \rrp_n :    &=  \begin{cases}
                                \rrppx{1}     &\quad \text{if}~ n = 1, \\
                                \max\big(\S_n\vee\big\{\rrppx{n}\big\}\big) &\quad \text{otherwise}, 
                            \end{cases} \\
            \S_n :      &=  \big\{\rrp_l + \rrp_k : l + k = n\big\}, \\ 
            \kappa(M) : &=  \bigg(\min \bigg\{n ~:~ \frac {\lrp_n}{n} < \frac{1}{2^{M+1}} ~\wedge\nonumber\\
                        &\hspace{2.75cm} \lrp_n-\frac{1}{2^{M+1}} < \frac{\rrp_n}{n}\bigg\}\bigg)^2,
        \end{align}
        for all \(n,M\in\NN\). Setting \(TM(\da,M) := \kappa(M)\) yields the required Turing machine. 
    \qed\end{proof}

\section{Computable Sequences of Rational Numbers versus Computable Sequences of Computable Numbers}        %
\label{comments}                                                                                            %

    In this section, we want to prove an additional result concerning the representation of number sequences. In particular, we will consider computable sequences \((x_n)_{n\in\NN}\) of computable numbers that additionally satisfy \(x_n\in\QQ\) for all \(n\in\NN\). At first sight, the difference to a representation of \((x_n)_{n\in\NN}\) in terms of \emph{Definition \ref{ber}} may not be apparent. However \((x_n)_{n\in\NN}\) being a computable sequence of rational numbers is a strictly stronger assertion that \((x_n)_{n\in\NN}\) being a computable sequence of computable numbers that additionally satisfies \(x_n\in\QQ\) for all \(n\in\NN\), as we will see in the following.

    \begin{Theorem}\label{thm:foo}
    There exist non-negative and computable sequences \((\rw_n)_{n\in\NN}\) and \((\ra_n)_{n\in\NN}\) of computable numbers that satisfy the following:
        \begin{enumerate}
            \item 
                The sequence \((\rw_n)_{n\in\NN}\) is a lower monotonic representation of a rational number and satisfies  \(\rw_n\in\QQ\) for all \(n\in\NN\), but is not a computable sequence of rational numbers.
            \item 
                The sequence \((\ra_n)_{n\in\NN}\) is a superadditive representation of a rational number and satisfies \(\ra_n\in\QQ\) for all \(n\in\NN\), but is not a computable sequence of rational numbers.
        \end{enumerate}
    \end{Theorem}\begin{proof}\mbox{}
        \begin{itemize}
            \item[1.] 
                We prove the statement by contradiction. Suppose all sequences \((\rw_n)_{n\in\NN}\) that satisfy all other requirements of \emph{Theorem \ref{thm:foo}} (Statement 1) are computable sequences of rational numbers. Let \(A\subset\NN\) be a recursively enumerable, non-recursive set and consider a Turing machine \(TM\) that \emph{accepts} \(A\), i.e. the domain of the corresponding recursive function \(f_{TM} : \NN \hookrightarrow \NN\) equals \(A\). In other words, the Turing machine \(TM\) \emph{halts} for input \(n\in\NN\) if and only if \(n\in A\) is satisfied. The recursive enumerability of \(A\) ensures the existence of such a Turing machine. For all \(n,m\in\NN\), define
                \begin{align}
                    l(n,m) :=\left\{    \begin{array}{ll}
                                            m_0     & ~\text{if, for input \(n \in \NN\),} \\
                                                    & ~\text{\(TM\) halts after}\\
                                                    & ~m_0 < m~\text{steps},\\
                                            m   & ~\text{otherwise}.
                                        \end{array}\right. 
                \end{align}
                Furthermore, for all \(n,m\in\NN\), define
                \begin{align}
                    r_{m,n} := \sum_{k=0}^{l(n,m)}\frac{1}{2^k} = 2-\frac{1}{2^{l(n,m)}}.
                \end{align}
                The sequence \((r_{m,n})_{m,n\in\NN}\) is a computable double sequence of rational numbers that is monotonically non-decreasing in \(m\) for all \(n\in\NN\). Assume \(n \in A\) holds true for some \(n\in\NN\). Then, there exists \(m_0\in\NN\) such that for all \(m\in\NN\) that satisfy \(m_0 \leq m\), we have \(l(n,m) = m_0\). Thus, if \(n\in\NN\) holds true, we have
                \begin{align}
                    \lim_{m\to\infty} r_{m,n} = 2 -\frac{1}{2^{m_0}} \in \QQ.
                \end{align}
                On the other hand, assume \(n\notin A\) holds true for some \(n\in\NN\). Then, for all \(m\in\NN\), we have \(l(n,m) = m\). Thus, if \(n\notin\NN\) holds true, we have
                \begin{align}
                    \lim_{m\to\infty} r_{m,n} = \lim_{m\to\infty} \left(2 -\frac{1}{2^{m}}\right) = 2 \in \QQ.
                \end{align}
                Define the sequence \((x_n)_{n\in\NN}\) via \(x_n := \lim_{m\to\infty} r_{m,n}\) for all \(n\in\NN\). Then, \((x_n)_{n\in\NN}\) is a sequence of numbers in \(\QQ\). By construction, we furthermore have \(| x_n - r_{m,n}| < \sfrac{1}{2^{M}}\) for all \(n,m,M\) that satisfy \(m \geq M + 1\). Thus, setting 
                \begin{align}
                    \kappa : \NN \rightarrow \NN, M \mapsto M + 1,   
                \end{align} 
                we conclude that the pair \(\big((r_{m,n})_{m,n\in\NN}, \kappa\big)\) is a standard description of the computable sequence \((x_n)_{n\in\NN}\) of computable numbers. For all \(n\in\NN\), define
                \begin{align}
                    \rw_n := 1 - \frac{1}{n} + \frac{1}{2n(n+1)}x_n.
                \end{align}
                Then, \((\rw_n)_{n\in\NN}\) is a non-negative, computable sequence of computable numbers attaining values in \(\QQ\). We have \(1 \leq x_n \leq 2\) and thus 
                \begin{align}
                    \big|1-\rw_n\big|=\left|\frac 1n-\frac 1{2(n(n+1))}x_n\right|<\frac{1}{n}   
                \end{align}
                for all \(n\in\NN\), i.e., the sequence \((\rw_n)_{n\in\NN}\) converges effectively towards \(1\) for all \(n\in\NN\). Furthermore, we have
                \begin{eqnarray}
                    \rw_{n+1}-\rw_n &=& \frac 1n - \frac 1{n+1}+ \frac 1{2(n+1)(n+2)}x_{n+1}\nonumber\\ &&- \frac 1{2n(n+1)}x_n\nonumber\\
                                    &>&  \frac 1{n(n+1)}- \frac 1{2n(n+1)}x_n\nonumber\\
                                    &\geq& \frac 1{n(n+1)}- \frac 1{n(n+1)}=0,
                \end{eqnarray}
                for all \(n\in\NN\), i.e., the sequence \((\rw_n)_{n\in\NN}\) is monotonically non-decreasing in \(n\). In summary, \((\rw_n)_{n\in\NN}\) is a non-negative, monotonically non-decreasing, computable sequence on computable numbers that satisfies \(\rw_n \in \QQ\) for all \(n\in \NN\) and converges effectively towards \(1\). By assumption, \((\rw_n)_{n\in\NN}\) is also a computable sequence of rational numbers. As \((\rw_n)_{n\in\NN}\) is furthermore non-negative, there exist recursive functions \(f_{\mathrm{nu}},f_{\mathrm{de}} :\NN \rightarrow \NN\) such that \(\rw_n = \sfrac{f_{\mathrm{nu}}(n)}{f_{\mathrm{de}}(n)}\) is satisfied for all \(n\in\NN\). Since comparisons on rational numbers are recursive operations, there exists a Turing machine \(TM' : \NN \rightarrow \NN\) that satisfies
                \begin{align}
                    TM'(n) = \left\{\begin{array}{ll}
                    1   & \text{if}~\frac{f_{\mathrm{nu}}(n)}{f_{\mathrm{de}}(n)} < 1 - \frac{1}{n+1}\\
                    0   & \text{otherwise}\end{array}\right.
                \end{align}
                for all \(n\in\NN\). We have \(TM'(n) = 1\) if and only if \(x_n < 2\). Consequently, \(TM'= \mathds {1}_A\) holds true, which contradicts the assumption of \(A\) being non-recursive.
            \item[2.] 
                Consider a sequence \((\rw_n)_{n\in\NN}\) that satisfies the requirements of \emph{Theorem \ref{thm:foo}} (Statement 1). For all \(n\in\NN\), define \(\ra_n := n\cdot \rw_n\). Then, the sequence \((\ra_n)_{n\in\NN}\) satisfies the requirements of \emph{Theorem \ref{thm:foo}} (Statement 2).
        \end{itemize}
    \qed\end{proof}

\section{Summary and Discussion}        %
\label{sec:discussions}                 %
    
    By proving the equalities \(\Sigma_1 = \Lp\) and \(\Pi_1 = \Lb\), we have established a full characterization of Fekete's lemma in view of the arithmetical hierarchy of real numbers and shown that a whole class of uncomputable problems can be represented by super- and subadditive sequences. As a corollary, we have seen that for a computable sequence \((\ra_n)_{n\in\NN}\) of computable numbers, superadditivity is \emph{not} a sufficient condition for the sequence \((\ra_n\cdot n^{-1})_{n\in\NN}\) to converge effectively. 

    If the conditions \((\lim_{n\to\infty} \ra_n\cdot n^{-1}) \in \RR_c\) and \(0 \leq \ra_n\) for all \(n\in\NN\) are satisfied additionally, then \((\ra_n\cdot n^{-1})_{n\in\NN}\) does converge effectively. Furthermore, if a suitable "converse" \((\la_n)_{n\in\NN}\) is available, a modulus of convergence for \((\ra_n\cdot n^{-1})_{n\in\NN}\) can be determined algorithmically. 

    On the other hand, given a non-negative, computable sequence \((x_n)_{n\in\NN}\) of computable numbers that converges towards a computable number, superadditivity of the sequence \((n\cdot x_n)_{n\in\NN}\) is a sufficient condition for effective convergence of \((x_n)_{n\in\NN}\). As pointed out in the beginning of \emph{Section \ref{effect}}, convergence towards a computable number is \emph{not} a sufficient condition for effective convergence in general. Whether the superadditivity-requirement on \((n\cdot x_n)_{n\in\NN}\) to ensure effective convergence of \((x_n)_{n\in\NN}\) can be further weakened is a pressing open question. 

    On a different note, our results imply that no constructive proof for Fekete's lemma can exist, and furthermore, no algorithm can exist that generally allows for the effective computation of the associated limit values. As already mentioned in \emph{Section \ref{sec:Intr}}, other methods of information theory have recently been investigated regarding their computability properties, and it has been shown that they are non-constructive. In particular, the associated problems cannot be solved algorithmically. Despite the fact that this behavior occurs in multiple cases, it does not seem to be sufficiently understood. It is present in Shannon's fundamental publication in information theory, where he proves the existence of capacity-achieving codes, but is unable to construct them effectively \cite{Sha48}. It was only recently shown in \cite{BSP21} that an effective construction of capacity achieving codes as a function of the channel is not possible. Even for several results that are significantly less complex in their structure, an analysis shows that the corresponding proofs are non-constructive. That is, the mathematical object of interest is proven to exist, but no algorithm to explicitly generate it is devised. The approach by Blahut and Arimoto for the numerical approximation of the capacity of discrete memoryless channels yields a particularly simple example \cite{Bla72, Ari72}. The computation of the capacity as a number is straightforward, due to the properties of the mutual information. Blahut and Arimoto describe a \emph{procedure} for the characterization of a corresponding optimizer, but the proof of their method is non-constructive \cite{Ari72,Csi74,CsTu84,Cov84} (in \cite{Cov84}, a related problem has been considered). 
    
    In a recent publication \cite{boche2023algorithmic} it was shown that no algorithm for Turing machines can exist for calculating or approximating the optimal input distribution. In \cite{lee2023computability}, the impossibility of calculating optimizations for other information-theoretic problems, including log-investment optimization from \cite{cover1984algorithm} was shown. A discussion of general approaches to the calculation of optimal values of functions and the calculation of optimization of functions can be found in \cite{lee2024computability}.
    
    For interesting weakening of the superadditivity requirement, reference is made to the work of Erdös \cite{BE51,BE52} and the current work \cite{FR18} by Füredi and Rushka.

Finally, the following problem should be mentioned in connection with the Fekete lemma, which was also initiated by Ning Cai. As already mentioned in the introduction, the Fekete lemma immediately yields the result that pessimistic and optimistic capacity are equal. The prerequisite for the Fekete lemma for the application in \cite{ABBN13}, \cite{BBS13} and \cite{WNB16} are always fulfilled in these works, since memoryless channels are considered there. However, memoryless channels do not necessarily lead to equal pessimistic and optimistic capacities. For example, for deterministic identification with finite dimensional output distributions but continuous input distributions in \cite{colomer2024deterministic}, it was shown that pessimistic and optimistic identification capacities assume different values.  


\section*{Acknowledgments}
The authors acknowledge the financial support by the Federal Ministry of Education and Research of Germany in the programme of “Souverän. Digital. Vernetzt.”. Joint project 6G-life, project identification number: 16KISK002. Holger Boche and Christian Deppe were supported in part by the Bundesministerium 
f\"ur Bildung und Forschung (BMBF) through the grants 16KISQ093 (QUIET), 16KIS1598K (QuaPhySI), 16KISQ077 (QDCamNetz), and 16KISR027K (QTREX). Christian Deppe was supported in part by the BMBF through grant 16KIS1005 (NEWCOM).
\bibliographystyle{spbasic}      
\bibliography{references}   
\end{document}